\documentclass[11pt]{article}

\usepackage[margin=1in]{geometry}
\usepackage{amsmath,amssymb,amsthm,mathtools}
\usepackage{aliascnt}
\usepackage{booktabs}
\usepackage{microtype}
\usepackage[table,dvipsnames]{xcolor}
\usepackage{hyperref}
\usepackage[nameinlink,capitalise,noabbrev]{cleveref}

\hypersetup{
    colorlinks=true,
    linkcolor=blue,
    citecolor=ForestGreen,
    urlcolor=blue!60!black
}

\numberwithin{equation}{section}

\newtheorem{theorem}{Theorem}[section]
\newaliascnt{lemma}{theorem}
\newtheorem{lemma}[lemma]{Lemma}
\aliascntresetthe{lemma}
\newaliascnt{proposition}{theorem}
\newtheorem{proposition}[proposition]{Proposition}
\aliascntresetthe{proposition}
\newaliascnt{corollary}{theorem}
\newtheorem{corollary}[corollary]{Corollary}
\aliascntresetthe{corollary}
\theoremstyle{definition}
\newaliascnt{definition}{theorem}
\newtheorem{definition}[definition]{Definition}
\aliascntresetthe{definition}
\newaliascnt{remark}{theorem}

\aliascntresetthe{remark}

\crefname{theorem}{Theorem}{Theorems}
\Crefname{theorem}{Theorem}{Theorems}
\crefname{lemma}{Lemma}{Lemmas}
\Crefname{lemma}{Lemma}{Lemmas}
\crefname{proposition}{Proposition}{Propositions}
\Crefname{proposition}{Proposition}{Propositions}
\crefname{corollary}{Corollary}{Corollaries}
\Crefname{corollary}{Corollary}{Corollaries}
\crefname{definition}{Definition}{Definitions}
\Crefname{definition}{Definition}{Definitions}
\crefname{remark}{Remark}{Remarks}
\Crefname{remark}{Remark}{Remarks}
\crefname{section}{Section}{Sections}
\Crefname{section}{Section}{Sections}
\crefname{equation}{Equation}{Equations}
\Crefname{equation}{Equation}{Equations}

\DeclarePairedDelimiter\norm{\lVert}{\rVert}
\DeclarePairedDelimiter\abs{\lvert}{\rvert}

\newcommand{\R}{\mathbb{R}}
\newcommand{\eps}{\varepsilon}

\newcommand{\indicator}{\mathbb{I}}
\newcommand{\Sphere}{\mathbb{S}}
\newcommand{\MAXIP}{\operatorname{MAX\text{-}IP}}
\newcommand{\Chamfer}{\operatorname{Chamfer}}
\newcommand{\rank}{\operatorname{rank}}
\newcommand{\arank}{\operatorname{rk}}
\newcommand{\adeg}{\deg}
\newcommand{\poly}{\operatorname{poly}}
\newcommand{\calV}{\mathcal{V}}
\newcommand{\calQ}{\mathcal{Q}}
\newcommand{\calX}{\mathcal{X}}

\title{Near-Optimal Dimension Lower Bounds for\\
Single-Vector Embeddings of Maximum Inner Product Similarity}
\author{
Rajesh Jayaram \footnote{Google Research. \texttt{rkjayaram@google.com}}
\and
Honghao Lin\footnote{Google Research, Carnegie Mellon University / Texas A\&M University. \texttt{honghaol3010@gmail.com}}
\and
Vahab Mirrokni\footnote{Google Research. \texttt{mirrokni@google.com}}
\and
David P. Woodruff\footnote{Google Research and Carnegie Mellon University. \texttt{dpwoodru@gmail.com }}
}
\date{}

\begin{document}
\maketitle

\begin{abstract}
Multi-vector embeddings represent items by point clouds and compare query and
document point clouds using Chamfer similarity, whereas single-vector
embeddings use ordinary inner products.  For singleton queries, Chamfer
becomes maximum inner product similarity (MAX-IP).  In our setting, MUVERA
gives dimension $m^{O(1/\eps^2)}$~\cite{dhulipala2024muvera}, whereas the previous
lower bound $(\eps^2m)^{\Omega(1/\eps)}$~\cite{jayaram2026expressive} left a
gap between $1/\eps$ and $1/\eps^2$ in the exponent of $m$.

We nearly close this gap.  For every fixed $\delta\in(0,1)$, there are
constants $A_\delta,c_\delta>0$ such that, for all sufficiently small
$\eps>0$ and every $m\ge(1/\eps)^{A_\delta}$, there exist unit query vectors
and document point clouds of at most $m$ unit vectors for which every
single-vector approximation of all pairwise MAX-IP values to additive error
$\eps$ has dimension
\[
    D \ge m^{c_\delta/\eps^{2-2\delta}}.
\]
This holds even for fully data-dependent representations chosen after seeing
the dataset.  It also applies to Chamfer because all queries are singletons.
Since $\delta$ can be arbitrarily small, the exponent approaches the
$O(1/\eps^2)$ dependence of the upper bound.

The proof combines Sherstov's pattern matrix method with polynomial-size,
constant-width DNF formulas computing functions of approximate degree
$\Omega(k^{1-\delta})$.  Uniform-width padding and a block encoding create an
$\Omega(\eps)$ gap.  A dummy coordinate then equalizes all false inputs,
yielding a unit-sphere MAX-IP matrix that is an exact two-valued affine image
of the DNF pattern matrix with gap at least $8\eps$.  This allows the
approximate-rank bound to apply.

The proof was first obtained using a fully automated Gemini-based agentic system developed internally at Google. The authors have verified the proof and edited it for clarity of presentation.
\end{abstract}

\section{Introduction}
\label{sec:introduction}

Vector representations are a basic interface between learned models and
large-scale retrieval systems.  In the traditional single-vector paradigm, a
query and a document are represented by vectors $q,x\in\R^d$, and their
similarity is the inner product $\langle q,x\rangle$.  This representation is
compact and is compatible with highly optimized maximum inner product search
data structures, but it compresses an entire data item into one vector.

Late-interaction models instead associate each data item with a collection of
vectors.  ColBERT~\cite{khattab2020colbert} and subsequent systems such as
ColBERTv2 and PLAID~\cite{santhanam2022colbertv2,santhanam2022plaid} retain
token-level information and score a query against a document by summing
per-query-vector maximum similarity interactions.  We use the query-normalized form
\begin{equation}
    \Chamfer(Q,X)
    := \frac{1}{|Q|}\sum_{q\in Q}\max_{x\in X}\langle q,x\rangle,
    \label{eq:chamfer-intro}
\end{equation}
which differs from the sum convention only by the query-dependent factor
$1/|Q|$ and therefore does not change document rankings for a fixed query.
This nonlinear aggregation retains fine-grained query--document interactions
before they are compressed into a single score.  Its practical success has
motivated a large body of work on multi-vector modeling and retrieval.  The
same fine-grained comparison also makes direct evaluation expensive: a naive
exact computation of \cref{eq:chamfer-intro} examines every query--document
vector pair.

The target dimension of a single-vector surrogate measures whether the
additional structure of a point cloud gives a genuine representational
advantage.  An $m$-point cloud in $\R^d$ uses $md$ scalar coordinates.  If
every finite family of such point clouds could always be replaced by single
vectors in dimension $O(md)$ while preserving all pairwise similarities, then
multi-vector representations would have no worst-case advantage in
representation size.
This leads to the central question of this work.
\begin{center}
\emph{How large must $D$ be for single-vector inner products to approximate
Chamfer similarity?}
\end{center}

More formally, consider finite collections of query and document point clouds,
each containing at most $m$ unit vectors.  We ask for the smallest $D$ such
that one can assign vectors $a_Q,b_X\in\R^D$ satisfying
\begin{equation}
    \abs*{\langle a_Q,b_X\rangle-\Chamfer(Q,X)}\le\eps
    \qquad\text{for every query }Q\text{ and document }X.
    \label{eq:intro-embedding-question}
\end{equation}
The assignments may depend on the entire dataset; a lower bound in this
data-dependent model also applies to oblivious or efficiently computable
embeddings.

MUVERA~\cite{dhulipala2024muvera} gives randomized oblivious assignments that,
for a fixed query--document pair, achieve additive error $\eps$ with failure
probability at most $\rho$ in dimension
\begin{equation}
    D=\left(\frac{m}{\eps}\right)^{O(1/\eps^2)}
      \log(1/\rho)\log(d/\rho),
    \label{eq:muvera-upper-intro}
\end{equation}
where constant-factor differences in the definition of $m$ are absorbed by
the asymptotic notation.  For finite datasets, one may choose $\rho$ small
enough, apply a union bound over all query--document pairs, and fix a successful
realization.  In the regime relevant below, where $d$ and $1/\eps$ are bounded
by fixed powers of $m$ and the dataset size is at most $2^{\poly(m)}$, the
resulting deterministic existence bound is still $D=m^{O(1/\eps^2)}$.  Thus
bounded finite datasets do admit single-vector surrogates, but the required
dimension can be much larger than the original point-cloud representation.

The first dimension lower bound against arbitrary data-dependent
representations was obtained in~\cite{jayaram2026expressive}.  For
$\eps\in(0,1/12)$ and $m\ge1/\eps^2$, that work constructed singleton queries
and document point clouds whose MAX-IP matrix requires single-vector dimension
$(\eps^2m)^{\Omega(1/\eps)}$ to approximate.  Since
\[
    \Chamfer(\{q\},X)=\MAXIP(q,X)
    =\max_{x\in X}\langle q,x\rangle,
\]
the construction already showed that the obstruction is present in the
simplest nonlinear component of Chamfer similarity.  The proof realized the
pattern matrix of the $k$-bit NAND function with
$k=\Theta(1/\eps^2)$.  The approximate degree
$\Theta(\sqrt{k})=\Theta(1/\eps)$ of NAND yielded the exponent
$\Theta(1/\eps)$, leaving open the gap to the $O(1/\eps^2)$ exponent in
\cref{eq:muvera-upper-intro}~\cite{jayaram2026expressive}.

\subsection{Our result}
\label{subsec:our-result}

We nearly close this gap already for MAX-IP, and hence also for Chamfer
similarity.  Our main result is the following.

\begin{theorem}[Main theorem]
\label{thm:main}
Fix any constant $\delta\in(0,1)$.  There exist constants
$A_\delta,c_\delta>0$ and $\eps_\delta\in(0,1/12)$ such that the following
holds for every $\eps\in(0,\eps_\delta)$ and every integer
$m\ge (1/\eps)^{A_\delta}$.

There are an ambient dimension $d\le m^{O_\delta(1)}$, a finite set of unit
query vectors $\calQ\subset\Sphere^{d-1}$, and a finite family $\calX$ of
nonempty document point clouds $X\subset\Sphere^{d-1}$ with
$1\le |X|\le m$, such that any maps
\[
    \Phi_Q:\calQ\to\R^D,
    \qquad
    \Phi_X:\calX\to\R^D
\]
satisfying
\begin{equation}
    \abs*{
      \langle\Phi_Q(q),\Phi_X(X)\rangle
      -\max_{x\in X}\langle q,x\rangle
    }
    \le \eps
    \qquad
    (q\in\calQ,\ X\in\calX)
    \label{eq:main-embedding-guarantee}
\end{equation}
must have
\[
    D\ge m^{c_\delta/\eps^{2-2\delta}}.
\]
Moreover,
\[
    |\calQ|,|\calX|\le 2^{m^{O_\delta(1)}}.
\]
\end{theorem}

Because $\delta\in(0,1)$ may be chosen arbitrarily small, the exponent
$2-2\delta$ can be made arbitrarily close to $2$.  The threshold exponent
$A_\delta$ is a constant that depends on $\delta$; we do not optimize this
dependence.

The comparison with the known bounds is summarized below.  In the parameter
regime of the theorem, for datasets with at most $2^{\poly(m)}$ queries and
documents, the MUVERA failure probability can be chosen small enough to fix one
realization that succeeds for every pair~\cite{dhulipala2024muvera}.
\begin{center}
\begin{tabular}{@{}lll@{}}
\toprule
Result & Dimension & Scope \\
\midrule
MUVERA~\cite{dhulipala2024muvera}
  & $m^{O(1/\eps^2)}$ & Chamfer upper bound \\
Previous lower bound~\cite{jayaram2026expressive}
  & $(\eps^2m)^{\Omega(1/\eps)}$
  & MAX-IP, $m\ge1/\eps^2$ \\
\rowcolor{blue!15}
This work
  & $m^{\Omega_\delta(1/\eps^{2-2\delta})}$
  & MAX-IP, $m\ge(1/\eps)^{A_\delta}$ \\[-0.35ex]
\bottomrule
\end{tabular}
\end{center}

\subsection{Technical overview}
\label{subsec:overview}

Any single-vector representation in dimension $D$ produces a score matrix of
rank at most $D$.  We therefore construct a MAX-IP matrix $M$ whose
$\eps$-approximate rank is large.  The construction realizes $M$ as an affine
image of a pattern matrix, allowing us to invoke Sherstov's pattern matrix
method~\cite{sherstov2008pattern}.

Let $h:\{0,1\}^k\to\{0,1\}$ be a Boolean function and let
$h_\pm=2h-1$ be its sign version.  For an integer $N$ divisible by $k$, index
the rows of the transposed pattern matrix by $(V,w)$, where $V$ selects one
coordinate from each of $k$ equal blocks of $[N]$ and $w\in\{0,1\}^k$, and
index its columns by $y\in\{0,1\}^N$.  Its entries are
\[
    F_{(V,w),y}=h_\pm(w\oplus y|_V).
\]
If $h_\pm$ has large approximate degree, then every constant-error entrywise
approximation to $F$ has large rank; quantitatively, the rank grows as
$(N/k)^{\Omega(\adeg(h_\pm))}$.

The previous MAX-IP lower bound used NAND for
$h$~\cite{jayaram2026expressive}.  We instead use Sherstov's construction of a
polynomial-size, constant-width DNF with approximate degree
$\Omega(k^{1-\delta})$~\cite{sherstov2022dnf}; its sign version inherits the
same asymptotic lower bound.  With $k=\Theta_\delta(1/\eps^2)$, the pattern
matrix has approximate-rank exponent
$\Omega_\delta(1/\eps^{2-2\delta})$.  The remaining task is to realize this
matrix by unit vectors and point clouds of controlled size while preserving
an $\Omega(\eps)$ gap between the two Boolean outputs.

Write the DNF as an OR of $s=k^{O_\delta(1)}$ terms, each containing at most
$W=O_\delta(1)$ literals.  We first make the width exact: a shorter term $T$
can be replaced by $(T\wedge z_i)\vee(T\wedge\neg z_i)$ using a variable
absent from $T$, and the operation is repeated until every term has exactly
$W$ distinct literals.  This padding preserves the function, and hence its
approximate degree, while increasing the number of terms by only a
$\delta$-dependent constant factor.  Exact width ensures that every fully
satisfied term produces the same inner product.

Next write $N=kB$ and partition the coordinates into $k$ blocks of size $B$.
A query indexed by $(V,w)$ chooses one active coordinate in every block and
records the corresponding bit of $w$ in either a positive or a negative copy
of that coordinate.  A document point cloud indexed by $y$ contains, for each
DNF term, all $B^W$ alignments of its literals with coordinates in their
corresponding blocks.  The resulting query and term vectors are normalized,
and their inner product counts the literals that are simultaneously aligned
and satisfied by $w\oplus y|_V$.  Maximizing over alignments therefore yields
\[
  \frac{W}{\sqrt{kW}}
  \quad\text{if the DNF is true},
  \qquad
  \text{at most }\frac{W-1}{\sqrt{kW}}
  \quad\text{if it is false}.
\]
The raw gap is $1/\sqrt{kW}=\Theta_\delta(\eps)$.  Including the dummy vector
introduced below, a document contains at most $sB^W+1$ vectors.  Under the
theorem's lower bound on $m$, we have $m\ge s^2+1$ and may choose
\[
    B=\left\lfloor\left(\frac{m-1}{s}\right)^{1/W}\right\rfloor
      =m^{\Omega_\delta(1)}.
\]

The false case still presents a problem: its maximum may depend on the input,
whereas the pattern matrix has only two values.  We resolve this with a
one-coordinate lift.  After scaling the original query, we append a fixed
coordinate $\eta$ and add the corresponding unit basis vector to every
document point cloud; all term vectors have zero in the new coordinate.  This
dummy vector places a floor at $\eta$.  The scale is chosen so that, on a
false input, every term vector has inner product at most $\eta$, whereas on a
true input some term vector has inner product $v_1>\eta$.  Thus false instances
have MAX-IP exactly $\eta$, true instances have MAX-IP exactly $v_1$, and all
vectors remain on the unit sphere.  Consequently, the similarity matrix $M$
satisfies the exact affine identity
\[
    M=\frac{v_1+\eta}{2}J
      +\frac{v_1-\eta}{2}F,
\]
where $J$ is the all-ones matrix.

Our choice of $k$ guarantees $v_1-\eta\ge8\eps$.  Hence a rank-$D$
$\eps$-approximation to $M$ becomes a $1/4$-approximation to $F$ after an
affine transformation of rank at most $D+1$.  Sherstov's rank lower bound then
gives
\[
    D+1\ge B^{\Omega(k^{1-\delta})}
      =m^{\Omega_\delta(1/\eps^{2-2\delta})},
\]
which is the claimed dimension lower bound.

\subsection{Related work}
\label{subsec:related-work}

ColBERT introduced the late-interaction paradigm, representing queries and
documents by token-level vectors and scoring them by the Chamfer similarity~\cite{khattab2020colbert}.  Related token-level architectures
include COIL, which combines contextual token representations with exact
lexical matching through inverted lists~\cite{gao2021coil}; ColBERTer, which
reduces storage by retaining selected whole-word
representations~\cite{hofstatter2022colberter}; and AligneR, which learns sparse
pairwise token alignments and per-token saliences~\cite{qian2022sparse}.
ColBERTv2 combined residual compression with denoised
supervision~\cite{santhanam2022colbertv2}, whereas PLAID introduced centroid
interaction and pruning in an optimized retrieval
engine~\cite{santhanam2022plaid}.  Specialized multi-vector search engines
include DESSERT for vector-set search~\cite{engels2023dessert}, EMVB with
bit-vector prefiltering~\cite{nardini2024bitvectors}, and WARP for XTR-based
late-interaction retrieval~\cite{scheerer2025warp}.  In
multimodal retrieval, ColPali applies late interaction to multi-vector
representations of document-page images and introduces the ViDoRe
benchmark~\cite{faysse2025colpali}.

Chamfer objectives have also been studied algorithmically outside information
retrieval.  For the distance variant, Bakshi et al. gave a near-linear-time
multiplicative approximation for two point sets~\cite{bakshi2023chamfer}; other
work considers Chamfer distance under translation~\cite{halevi2026translation}
and in the fully dynamic setting~\cite{goranci2025dynamic}.  These results concern
distance computation, rather than single-vector approximation of a prescribed
similarity matrix.  For the normalized inner-product similarity studied here,
MUVERA constructs randomized asymmetric fixed-dimensional encodings for sets
of unit vectors whose inner product gives an additive approximation, thereby
reducing multi-vector retrieval to single-vector maximum inner product
search~\cite{dhulipala2024muvera}.

Recent work has identified limitations of single-vector retrieval under other
objectives.  Weller et al. bound the number of top-$k$ document subsets
realizable by fixed-dimensional inner-product retrieval and introduce the
LIMIT benchmark~\cite{weller2025limitations}.  A contemporaneous study observes
that dimensionality alone does not explain failures on LIMIT-like data,
identifies domain shift and score--relevance misalignment as major factors,
and compares single- and multi-vector retrieval through toy models of the
``drowning in documents'' effect~\cite{agarwal2026strengths}.  Lakshman et al.
study a different notion, stability of high-dimensional near-neighbor search,
and give conditions under which multi-vector search under the Chamfer distance
inherits single-vector stability, whereas average pooling need
not~\cite{lakshman2025stability}.  These works
concern ranking realizability, empirical generalization, or stability rather
than entrywise additive approximation of a prescribed Chamfer or MAX-IP score
matrix.

The work most directly preceding ours proves such an additive-approximation
lower bound against arbitrary data-dependent single-vector representations.
Using the pattern matrix of NAND, it obtains
$D=(\eps^2m)^{\Omega(1/\eps)}$ for $m\ge1/\eps^2$
\cite{jayaram2026expressive}.  Our proof draws on Sherstov's pattern matrix
method, which lower-bounds the approximate rank of the pattern matrix
associated with a Boolean function in terms of its approximate degree
\cite{sherstov2008pattern}, and on his construction, for every fixed
$\delta>0$, of polynomial-size, $O_\delta(1)$-width DNF and CNF formulas with
approximate degree $\Omega(n^{1-\delta})$~\cite{sherstov2022dnf}.  The new step
needed to combine these ingredients is an exact affine realization of the
resulting pattern matrix by unit-sphere MAX-IP, with a controlled additive gap
and at most $m$ vectors in every document point cloud.


\section{Preliminaries}
\label{sec:preliminaries}

For an integer $r\ge1$, write $[r]=\{1,\ldots,r\}$.  All logarithms are base
two unless stated otherwise.  For $A\in\R^{s\times t}$, let
$\norm{A}_\infty=\max_{i,j}|A_{ij}|$.  The unit sphere in $\R^d$ is denoted by
$\Sphere^{d-1}$.

\subsection{MAX-IP, Chamfer similarity, and embeddings}

For a vector $q\in\R^d$ and a nonempty finite point cloud $X\subset\R^d$,
define
\[
    \MAXIP(q,X):=\max_{x\in X}\langle q,x\rangle.
\]
For nonempty point clouds $Q,X\subset\R^d$, the normalized Chamfer similarity
is
\[
    \Chamfer(Q,X)
    :=\frac1{|Q|}\sum_{q\in Q}\MAXIP(q,X).
\]
The asymmetry is intentional: $Q$ is the query and $X$ is the document.
For singleton queries, $\Chamfer(\{q\},X)=\MAXIP(q,X)$; hence MAX-IP lower
bounds transfer immediately to normalized Chamfer similarity.

\begin{definition}[Data-dependent MAX-IP embedding]
\label{def:data-dependent-embedding}
Let $\calQ\subset\R^d$ be a finite set of query vectors and let $\calX$ be a
finite family of nonempty finite point clouds in $\R^d$.  A data-dependent
$\eps$-embedding of MAX-IP into inner product in dimension $D$ consists of
arbitrary maps $\Phi_Q:\calQ\to\R^D$ and $\Phi_X:\calX\to\R^D$ such that
\[
    \abs*{\langle\Phi_Q(q),\Phi_X(X)\rangle-\MAXIP(q,X)}
    \le\eps
\]
for every $q\in\calQ$ and $X\in\calX$.
\end{definition}

The maps in \cref{def:data-dependent-embedding} may depend on the full
collections $\calQ$ and $\calX$ and need not be computable or oblivious.  A
lower bound in this model therefore also applies to any more structured
embedding scheme.

\subsection{Approximate rank}

\begin{definition}[Approximate rank]
For a real matrix $M$ and $\gamma\ge0$, its $\gamma$-approximate rank is
\[
    \arank_\gamma(M)
    :=\min\bigl\{\rank(A):\norm{A-M}_\infty\le\gamma\bigr\}.
\]
\end{definition}

Approximate rank is exactly the algebraic quantity measured by
data-dependent single-vector embeddings.

\begin{proposition}[Embeddings and approximate rank]
\label{prop:embedding-rank}
Let $M\in\R^{s\times t}$.  Allowing $D=0$, the minimum dimension $D$ for which
there exist vectors $a_1,\ldots,a_s,b_1,\ldots,b_t\in\R^D$ satisfying
\[
    |\langle a_i,b_j\rangle-M_{ij}|\le\eps
    \qquad(i\in[s],\ j\in[t])
\]
is $\arank_\eps(M)$.
\end{proposition}

\begin{proof}
The matrix $A$ with entries $A_{ij}=\langle a_i,b_j\rangle$ factors as
$UV^\top$ and has rank at most $D$, proving
$\arank_\eps(M)\le D$.  Conversely, if $A$ has rank $r$ and
$\norm{A-M}_\infty\le\eps$, choose a rank factorization $A=UV^\top$ with
$U\in\R^{s\times r}$ and $V\in\R^{t\times r}$.  The rows of $U$ and $V$ give
the required vectors in $\R^r$.
\end{proof}

Applied to the MAX-IP score matrix
$M_{q,X}:=\MAXIP(q,X)$ indexed by $\calQ\times\calX$, the proposition
identifies the optimal data-dependent embedding dimension with
$\arank_\eps(M)$.

\subsection{Approximate degree and pattern matrices}

\begin{definition}[Approximate degree]
\label{def:approximate-degree}
For a function $f:\{0,1\}^k\to\R$ and $\alpha\ge0$, its
$\alpha$-approximate degree $\adeg_\alpha(f)$ is the least degree of a real
multivariate polynomial $p$ such that
$|p(z)-f(z)|\le\alpha$ for every $z\in\{0,1\}^k$.
\end{definition}

Throughout, $\oplus$ denotes coordinatewise XOR.

Fix integers $N,k\ge1$ such that $k$ divides $N$, and partition $[N]$ into
$k$ consecutive blocks $I_1,\ldots,I_k$, each of size $N/k$.  Let
\begin{equation}
    \calV(N,k)
    :=\bigl\{V\subset[N]:|V\cap I_i|=1\text{ for every }i\in[k]\bigr\}.
    \label{eq:block-selector-family}
\end{equation}
If $V=\{v_1<\cdots<v_k\}\in\calV(N,k)$, then $v_i\in I_i$ for every
$i\in[k]$.  For $y\in\{0,1\}^N$, write
$y|_V=(y_{v_1},\ldots,y_{v_k})$.

\begin{definition}[Pattern matrix]
\label{def:pattern-matrix}
Let $f:\{0,1\}^k\to\R$.  The $(N,k,f)$-pattern matrix is the matrix whose rows
are indexed by $y\in\{0,1\}^N$, whose columns are indexed by
$(V,w)\in\calV(N,k)\times\{0,1\}^k$, and whose entries are
\[
    F_{y,(V,w)}=f(w\oplus y|_V).
\]
\end{definition}

Our MAX-IP matrix later uses $(V,w)$ as row indices and $y$ as column indices,
and is therefore the transpose of the matrix in
\cref{def:pattern-matrix}.  This is immaterial because
$\arank_\gamma(F)=\arank_\gamma(F^\top)$.  We use the following
approximate-rank consequence, which is Theorem~8.1 of
Sherstov~\cite{sherstov2008pattern}.
\begin{theorem}[Pattern matrix method]
\label{thm:pattern-rank}
Let $F$ be the $(N,k,f)$-pattern matrix of a function
$f:\{0,1\}^k\to\{-1,1\}$.  For every $\alpha\in[0,1)$ and every
$\gamma\in[0,\alpha]$,
\[
    \arank_\gamma(F)
    \ge
    \left(\frac{\alpha-\gamma}{1+\gamma}\right)^2
    \left(\frac Nk\right)^{\adeg_\alpha(f)}.
\]
\end{theorem}

\section{High-Approximate-Degree DNFs of Exact Width}
\label{sec:hard-dnf}

For an input $z\in\{0,1\}^k$, a \emph{literal} is either a positive literal
$z_j$ or a negative literal $\neg z_j:=1-z_j$; the literal is
\emph{satisfied} when it evaluates to one.  A \emph{term} is a conjunction of
literals.  Thus a disjunctive normal form (DNF) representation of a Boolean
function $f$ has the form
\[
    f(z)=T_1(z)\vee\cdots\vee T_s(z),
    \qquad
    T_i(z)=\bigwedge_{\ell\in L_i}\ell(z),
\]
where $L_i$ is the collection of literals in the $i$th term.  In particular,
$T_i(z)=1$ exactly when every literal in $L_i$ is satisfied, and $f(z)=1$
exactly when at least one term is satisfied.  The \emph{size} of the DNF is
the number $s$ of terms, and its \emph{width} is
$\max_{i\in[s]}|L_i|$, counting distinct literals.  The DNF has
\emph{exact width} $W$ if $|L_i|=W$ for every $i$.  A term contains a
\emph{contradictory pair} if $L_i$ contains both $z_j$ and $\neg z_j$ for some
$j$.

Our base functions come from Sherstov's near-linear lower bound on the
approximate degree of DNF formulas.  The geometric construction in
\cref{sec:geometric-construction} requires every term to consist of the same
number of literals on distinct variables: this makes satisfying inputs attain
a common score, whereas falsifying inputs miss at least one of those literals.
Recall from \cref{def:approximate-degree} that $\adeg_\alpha(f)$ is the least
degree of a real polynomial that approximates $f$ pointwise to error at most
$\alpha$.

\begin{theorem}[Sherstov~\cite{sherstov2022dnf}]
\label{thm:sherstov-dnf}
For every constant $\delta\in(0,1)$, there are constants
$W_0=W_0(\delta)\ge1$, $C_0=C_0(\delta)\ge1$, $a_0=a_0(\delta)>0$, and
$k_0=k_0(\delta)$ such that for every integer $k\ge k_0$ there is an explicit
function $f:\{0,1\}^k\to\{0,1\}$ with
\[
    \adeg_{1/3}(f)\ge a_0k^{1-\delta}
\]
that is computable by a DNF of width at most $W_0$ and size at most $k^{C_0}$.
\end{theorem}

Sherstov's theorem guarantees width at most $W_0$.  The next elementary lemma
makes the width exact at a constant-factor cost in the number of terms.

\begin{lemma}[Uniform-width padding]
\label{lem:uniform-width-padding}
Suppose $f:\{0,1\}^k\to\{0,1\}$ is represented by a DNF with $s$ terms, each
containing at most $W_0$ literals involving distinct variables.  For any
integer $W$ with $W_0\le W\le k$, the same function has a DNF in which every
term consists of exactly $W$ literals involving $W$ distinct variables, no
term contains a contradictory pair, and the number of terms is at most $s2^W$.
\end{lemma}

\begin{proof}
Consider a term $T$ containing $w<W$ literals on $w$ distinct variables.
Since $w<W\le k$, some variable $z_j$ does not occur in $T$.  The identity
\[
    T\equiv(T\wedge z_j)\vee(T\wedge\neg z_j)
\]
replaces $T$ by two terms, each containing $w+1$ literals on $w+1$ distinct
variables, without changing the represented function or introducing a
contradictory pair.  Repeating the operation on every descendant of $T$
produces $2^{W-w}$ terms, each containing exactly $W$ literals on $W$ distinct
variables and no contradictory pair.  Applying this procedure independently
to all original terms produces at most $\sum_T2^{W-|T|}\le s2^W$ terms.
\end{proof}

\begin{corollary}[Hard exact-width DNF]
\label{cor:hard-uniform-dnf}
For every constant $\delta\in(0,1)$, there are constants
$W=W(\delta)\ge1$, $C=C(\delta)\ge1$, $a=a(\delta)>0$, and
$k_0'=k_0'(\delta)$ such that for every $k\ge k_0'$ there is an explicit
function $f:\{0,1\}^k\to\{0,1\}$ satisfying
\[
    \adeg_{1/3}(f)\ge ak^{1-\delta}
\]
and represented by a DNF with $1\le s\le k^C$ terms, each consisting of
exactly $W$ literals on $W$ distinct variables.  In particular, no term
contains a contradictory pair.
\end{corollary}

\begin{proof}
Apply \cref{thm:sherstov-dnf}, delete contradictory terms and repeated
literals, set $W=\lceil W_0\rceil$ and $a=a_0$, and assume
$k\ge\max\{k_0,W\}$.  Since $\adeg_{1/3}(f)>0$, the function $f$ is
nonconstant, so at least one term remains after these deletions.  By
\cref{lem:uniform-width-padding}, the number of terms after padding is at most
$2^Wk^{C_0}$.  Since $W$ depends only on
$\delta$, enlarging $C$ and the threshold $k_0'$ gives
$2^Wk^{C_0}\le k^C$.  The padding procedure is explicit.
\end{proof}

\section{Realizing Exact-Width DNFs by MAX-IP}
\label{sec:geometric-construction}

Let $f:\{0,1\}^k\to\{0,1\}$ be represented by an exact-width DNF
\[
    f=T_1\vee\cdots\vee T_s,
\]
with $s\ge1$ terms, each consisting of exactly $W$ literals on $W$ distinct
variables.  In particular, $1\le W\le k$.  Fix an integer
$B\ge1$ and set $N:=kB$.  Partition $[N]$ into $k$ consecutive blocks
$I_1,\ldots,I_k$ of size $B$, and use the selector family $\calV(N,k)$ from
\cref{eq:block-selector-family}.  Thus $|\calV(N,k)|=B^k$.

The construction uses one document vector for every DNF term and every choice
of one coordinate from each block relevant to that term.  Listing all $B^W$
alignments ensures that, for every selector $V\in\calV(N,k)$, the alignment
chosen by $V$ is present in the document.  We first build vectors whose inner
product counts aligned satisfied literals, and then add one dummy coordinate
to obtain an exact two-valued MAX-IP score.

For term $T_i$, let
\[
    1\le u_{i,1}<\cdots<u_{i,W}\le k
\]
be the indices of its variables.  Let $c_{i,j}\in\{0,1\}$ be the value that
the variable $z_{u_{i,j}}$ must take to satisfy the $j$th literal.  Thus the
literal is positive when $c_{i,j}=1$ and negated when $c_{i,j}=0$.

\subsection{Literal-counting construction}

We first work in $\R^{2N}$, with a positive and a negative coordinate for each
element of $[N]$.  Let $e_1,\ldots,e_{2N}$ denote the standard basis.

\begin{definition}[Literal-counting vectors]
\label{def:literal-counting-vectors}
Fix $V=\{v_1<\cdots<v_k\}\in\calV(N,k)$ and $w\in\{0,1\}^k$.  Define
\begin{equation}
    q_{V,w}
    :=\frac1{\sqrt{k}}
      \sum_{r=1}^k
      \bigl(w_re_{v_r}+(1-w_r)e_{N+v_r}\bigr).
    \label{eq:unpadded-query}
\end{equation}

For $y\in\{0,1\}^N$, a term $T_i$, and an alignment
\[
    \tau=(\tau(1),\ldots,\tau(W))
    \in I_{u_{i,1}}\times\cdots\times I_{u_{i,W}},
\]
set
\[
    b_{i,j}(y,\tau):=y_{\tau(j)}\oplus c_{i,j}
\]
and define
\begin{equation}
    p_{i,\tau}(y)
    :=\frac1{\sqrt W}
      \sum_{j=1}^W
      \bigl(
        b_{i,j}(y,\tau)e_{\tau(j)}
        +(1-b_{i,j}(y,\tau))e_{N+\tau(j)}
      \bigr).
    \label{eq:unpadded-document}
\end{equation}
\end{definition}

Both vectors in \cref{def:literal-counting-vectors} have unit norm: their nonzero
coordinates are distinct, and they have respectively $k$ entries of magnitude
$1/\sqrt{k}$ and $W$ entries of magnitude $1/\sqrt W$.

The choice between the positive and negative copies of $\tau(j)$ is arranged
so that overlap with the query records whether the $j$th literal is satisfied
by $z=w\oplus y|_V$.  The following identity makes this interpretation exact.

\begin{lemma}[Literal-counting identity]
\label{lem:literal-counting}
Let $z=w\oplus y|_V$.  For every $i$ and $\tau$,
\begin{equation}
    \langle q_{V,w},p_{i,\tau}(y)\rangle
    =\frac1{\sqrt{kW}}
      \sum_{j=1}^W
      \indicator\bigl[
        \tau(j)=v_{u_{i,j}}
        \text{ and }
        z_{u_{i,j}}=c_{i,j}
      \bigr].
    \label{eq:literal-counting}
\end{equation}
\end{lemma}

\begin{proof}
Fix $j\in[W]$.  In the block indexed by $u_{i,j}$, the query is supported at
the positive or negative copy of $v_{u_{i,j}}$, while the document vector is
supported at the corresponding copy of $\tau(j)$.  If
$\tau(j)\ne v_{u_{i,j}}$, the supports are disjoint and the contribution is
zero.

Suppose $\tau(j)=v_{u_{i,j}}$.  The contribution, before the common factor
$1/\sqrt{kW}$, is
\[
    w_{u_{i,j}}b_{i,j}
    +(1-w_{u_{i,j}})(1-b_{i,j}),
\]
which equals one exactly when $w_{u_{i,j}}=b_{i,j}$.  Using the definition of
$b_{i,j}$ and the equality $\tau(j)=v_{u_{i,j}}$, this condition is
equivalent to
\[
    w_{u_{i,j}}\oplus y_{v_{u_{i,j}}}=c_{i,j},
\]
or $z_{u_{i,j}}=c_{i,j}$.  Summing the contributions proves
\cref{eq:literal-counting}.
\end{proof}

For fixed $y$, let
\begin{equation}
    P(y)
    :=\bigl\{
      p_{i,\tau}(y):
      i\in[s],\
      \tau\in I_{u_{i,1}}\times\cdots\times I_{u_{i,W}}
    \bigr\}.
    \label{eq:unpadded-cloud}
\end{equation}
There are at most $sB^W$ distinct vectors in $P(y)$.

\begin{lemma}[Unpadded MAX-IP gap]
\label{lem:unpadded-gap}
For $z=w\oplus y|_V$,
\[
    \MAXIP(q_{V,w},P(y))
    \begin{cases}
      =\dfrac{W}{\sqrt{kW}},& f(z)=1,\\[1.2ex]
      \le\dfrac{W-1}{\sqrt{kW}},& f(z)=0.
    \end{cases}
\]
\end{lemma}

\begin{proof}
For a fixed term $T_i$, the alignment
$\tau(j)=v_{u_{i,j}}$ for every $j$ is available in the point cloud.  By
\cref{lem:literal-counting}, it produces one contribution
$1/\sqrt{kW}$ for each literal of $T_i$ satisfied by $z$.  No other alignment
can produce a larger contribution, because misaligned coordinates contribute
zero.

If $f(z)=1$, some exact-width term has all $W$ literals satisfied and achieves
$W/\sqrt{kW}$.  No vector can exceed this value.  If $f(z)=0$, every term has
at most $W-1$ satisfied literals, giving the stated upper bound.
\end{proof}

\subsection{An exact two-valued realization}

The false case of \cref{lem:unpadded-gap} is only an upper bound.  To obtain an
exact two-valued MAX-IP score, define
\begin{equation}
    \eta:=\frac{W-1}{\sqrt{kW+(W-1)^2}},
    \qquad
    v_1:=\frac{W}{\sqrt{kW+(W-1)^2}}.
    \label{eq:eta-v-one}
\end{equation}
Lift the query and document vectors to $\R^{2N+1}$ by
\begin{equation}
    q'_{V,w}:=(\sqrt{1-\eta^2}\,q_{V,w},\eta),
    \qquad
    p'_{i,\tau}(y):=(p_{i,\tau}(y),0),
    \label{eq:padded-vectors}
\end{equation}
and let $p_0=(0,\ldots,0,1)$.  Finally set
\begin{equation}
    Y_y
    :=\bigl\{
      p'_{i,\tau}(y):
      i\in[s],\
      \tau\in I_{u_{i,1}}\times\cdots\times I_{u_{i,W}}
    \bigr\}\cup\{p_0\}.
    \label{eq:padded-cloud}
\end{equation}

\begin{proposition}[Two-valued MAX-IP realization]
\label{prop:two-valued-realization}
The query vectors $q'_{V,w}$ and every vector in $Y_y$ lie on the unit sphere
in $\R^{2kB+1}$, and $|Y_y|\le sB^W+1$.  For every $(V,w)$ and $y$, with
$z=w\oplus y|_V$,
\begin{equation}
    \MAXIP(q'_{V,w},Y_y)
    =
    \begin{cases}
      v_1,&f(z)=1,\\
      \eta,&f(z)=0.
    \end{cases}
    \label{eq:exact-two-values}
\end{equation}
Moreover, the two values are separated by
\begin{equation}
    v_1-\eta=\frac1{\sqrt{kW+(W-1)^2}}.
    \label{eq:two-valued-gap}
\end{equation}
\end{proposition}

\begin{proof}
The vectors in \cref{def:literal-counting-vectors} are unit vectors, and
\[
    \norm{q'_{V,w}}_2^2
    =(1-\eta^2)\norm{q_{V,w}}_2^2+\eta^2
    =1.
\]
Appending a zero preserves the norms of the document vectors, while $p_0$ is
also a unit vector.  The ambient dimension and the bound on $|Y_y|$ follow
directly from the construction.

The dummy vector gives
$\langle q'_{V,w},p_0\rangle=\eta$.  Every other document vector satisfies
\[
    \langle q'_{V,w},p'_{i,\tau}(y)\rangle
    =\sqrt{1-\eta^2}\,
      \langle q_{V,w},p_{i,\tau}(y)\rangle.
\]
From \cref{eq:eta-v-one},
\[
    \sqrt{1-\eta^2}
    =\frac{\sqrt{kW}}{\sqrt{kW+(W-1)^2}}.
\]
Therefore the false-case threshold in \cref{lem:unpadded-gap} scales exactly
to $\eta$, while the true-case value scales exactly to $v_1$.  If $f(z)=0$,
all term vectors have inner product at most $\eta$ and the dummy vector
attains $\eta$.  If $f(z)=1$, a term vector attains $v_1>\eta$ and no vector
exceeds it.  Finally, subtracting the two quantities in
\cref{eq:eta-v-one} gives \cref{eq:two-valued-gap}.
\end{proof}

\section{Approximate Rank of the MAX-IP Matrix}
\label{sec:matrix-reduction}

Retain the notation and construction of \cref{sec:geometric-construction}.
Index the queries by $(V,w)\in\calV(N,k)\times\{0,1\}^k$ and the documents by
$y\in\{0,1\}^N$.  Define the MAX-IP matrix
\begin{equation}
    M_{(V,w),y}:=\MAXIP(q'_{V,w},Y_y).
    \label{eq:maxip-matrix}
\end{equation}
Let $F\in\{-1,1\}^{(B^k2^k)\times2^N}$ be the transpose of the
$(N,k,f_\pm)$-pattern matrix, where
\begin{equation}
    f_\pm(z):=2f(z)-1.
    \label{eq:f-sign}
\end{equation}
Thus $F_{(V,w),y}=2f(w\oplus y|_V)-1$.

The two-valued realization makes $M$ an affine copy of $F$.  After rescaling,
an entrywise approximation to $M$ becomes an approximation to $F$, at the
cost of at most one in rank.  The pattern matrix method can then be applied
directly.

\subsection{Affine reduction to the pattern matrix}

\begin{proposition}[Affine reduction]
\label{prop:affine-reduction}
The matrices $M$ and $F$ satisfy
\begin{equation}
    M=\frac{v_1+\eta}{2}J+\frac{v_1-\eta}{2}F,
    \label{eq:affine-identity}
\end{equation}
where $J$ is the all-ones matrix of the same dimensions.  Moreover, let
$\rho>0$ satisfy $v_1-\eta\ge8\rho$.  If $\widetilde M$ has rank at most $D$
and $\norm{\widetilde M-M}_\infty\le\rho$, then
\begin{equation}
    \widetilde F
    :=\frac{2}{v_1-\eta}
      \left(\widetilde M-\frac{v_1+\eta}{2}J\right)
    \label{eq:shifted-approximation}
\end{equation}
has rank at most $D+1$ and satisfies
$\norm{\widetilde F-F}_\infty\le1/4$.
\end{proposition}

\begin{proof}
By \cref{prop:two-valued-realization}, an entry of $M$ is $v_1$ when the
corresponding entry of $F$ is $1$, and it is $\eta$ when that entry is $-1$.
This proves \cref{eq:affine-identity}.

Since $J$ has rank one, rank subadditivity gives
$\rank(\widetilde F)\le D+1$.  Combining
\cref{eq:affine-identity,eq:shifted-approximation},
\[
    \norm{\widetilde F-F}_\infty
    =\frac{2}{v_1-\eta}\norm{\widetilde M-M}_\infty
    \le\frac{2\rho}{8\rho}=\frac14.
\]
\end{proof}

\subsection{Rank lower bound}

\begin{lemma}[Approximate rank of the constructed MAX-IP matrix]
\label{lem:maxip-rank-lower-bound}
For every $\rho>0$ such that $v_1-\eta\ge8\rho$, the matrix $M$ in
\cref{eq:maxip-matrix} satisfies
\[
    \arank_\rho(M)
    \ge \frac1{225}B^{\adeg_{1/3}(f)}-1.
\]
\end{lemma}

\begin{proof}
The affine conversion in \cref{prop:affine-reduction} shows that
\[
    \arank_\rho(M)+1\ge\arank_{1/4}(F).
\]
The matrix $F$ is the transpose of the $(N,k,f_\pm)$-pattern matrix.  Apply
\cref{thm:pattern-rank} with $\alpha=1/3$ and $\gamma=1/4$.  Since
$N/k=B$,
\begin{equation}
    \arank_{1/4}(F)
    \ge
    \left(\frac{1/3-1/4}{1+1/4}\right)^2
    B^{\adeg_{1/3}(f_\pm)}
    =\frac1{225}B^{\adeg_{1/3}(f_\pm)}.
    \label{eq:pattern-rank-applied}
\end{equation}

If a polynomial approximates $f_\pm=2f-1$ to error $1/3$, then adding one and
dividing by two gives an approximation to $f$ with error $1/6$.  Conversely,
an error-$1/6$ approximation to $f$, after multiplying by two and subtracting
one, approximates $f_\pm$ to error $1/3$.  Hence
\[
    \adeg_{1/3}(f_\pm)
    =\adeg_{1/6}(f)
    \ge\adeg_{1/3}(f).
\]
Substitution into \cref{eq:pattern-rank-applied} proves the lemma.
\end{proof}

\section{Proof of the Main Theorem}
\label{sec:main-proof}

We now choose the parameters of the general construction and translate
\cref{lem:maxip-rank-lower-bound} into a power of $m$.

\begin{proof}[Proof of \cref{thm:main}]
Fix $\delta\in(0,1)$ and the constants $W,C,a,k_0'$ from
\cref{cor:hard-uniform-dnf}.  Choose $\eps_\delta\in(0,1/12)$ small enough
that for every $\eps<\eps_\delta$, the integer
\begin{equation}
    k
    :=\left\lfloor
       \frac{1}{64W\eps^2}-\frac{(W-1)^2}{W}
      \right\rfloor
    \label{eq:k-choice}
\end{equation}
satisfies
$k\ge\max\{W,k_0'\}$ and
\begin{equation}
    c_1(\delta)\eps^{-2}
    \le k\le
    c_2(\delta)\eps^{-2}
    \label{eq:k-two-sided}
\end{equation}
for positive constants $c_1(\delta),c_2(\delta)$.

Let $f$ be the exact-width DNF supplied by
\cref{cor:hard-uniform-dnf}, so that
\begin{equation}
    \adeg_{1/3}(f)\ge ak^{1-\delta}
    \qquad\text{and}\qquad
    s\le k^C,
    \label{eq:hard-dnf-properties}
\end{equation}
where $s$ is its number of terms.  The choice in \cref{eq:k-choice} gives
\begin{equation}
    kW+(W-1)^2\le\frac1{64\eps^2},
    \qquad\text{hence}\qquad
    v_1-\eta\ge8\eps
    \label{eq:gap-calibration}
\end{equation}
by \cref{eq:two-valued-gap}.

By \cref{eq:k-two-sided}, there is a constant $c_3(\delta)$ such that
$s\le c_3(\delta)\eps^{-2C}$.  Set
\[
    A_\delta:=4C+1.
\]
For all sufficiently small $\eps$, the assumption
$m\ge(1/\eps)^{A_\delta}$ implies
\begin{equation}
    m\ge s^2+1.
    \label{eq:m-dominates-s}
\end{equation}
Define
\begin{equation}
    B:=\left\lfloor\left(\frac{m-1}{s}\right)^{1/W}\right\rfloor,
    \qquad
    N:=kB.
    \label{eq:B-and-N}
\end{equation}
Then $B\ge1$, $sB^W+1\le m$, and
\[
    B
    \ge\left\lfloor(m-1)^{1/(2W)}\right\rfloor
    \ge m^{b_\delta}
\]
for all sufficiently large $m$, where one may take any fixed
$b_\delta<1/(2W)$ after decreasing $\eps_\delta$ if necessary.

\medskip\noindent\emph{The dataset.}
Apply the construction of \cref{sec:geometric-construction} with this
$f,k,s,W$, and $B$, and define
\[
    \calQ
    :=\bigl\{
      q'_{V,w}:V\in\calV(N,k),\ w\in\{0,1\}^k
    \bigr\}.
\]
The query vectors are distinct because their supports and their choices of
positive or negative coordinates determine $(V,w)$.  Let $\calX$ contain one
representative of each distinct point cloud among the $Y_y$, $y\in\{0,1\}^N$.
By \cref{prop:two-valued-realization}, all vectors lie on the unit sphere in ambient
dimension
\[
    d=2N+1=2kB+1\le2km+1\le m^{O_\delta(1)},
\]
where the final inequality follows from
$m\ge(1/\eps)^{A_\delta}$ and $k=O_\delta(1/\eps^2)$.  Every document contains
at most $sB^W+1\le m$ vectors.

\medskip\noindent\emph{Dimension lower bound.}
Let $M^\circ$ be the MAX-IP score matrix indexed by $\calQ\times\calX$.
The indexed matrix $M$ from \cref{eq:maxip-matrix} is obtained from $M^\circ$
by duplicating columns.  Hence
$\arank_\eps(M^\circ)=\arank_\eps(M)$: approximations restrict from the larger
matrix to the smaller one and extend in the other direction by duplicating
columns.  By \cref{prop:embedding-rank}, every data-dependent embedding of
$\calQ\times\calX$ in dimension $D$ therefore satisfies
\[
    D\ge\arank_\eps(M^\circ)=\arank_\eps(M).
\]
Combining \cref{eq:gap-calibration,lem:maxip-rank-lower-bound} with
\cref{eq:hard-dnf-properties},
\[
    D
    \ge\frac1{225}B^{ak^{1-\delta}}-1
    \ge\frac1{225}
       m^{ab_\delta k^{1-\delta}}-1.
\]
Using the lower bound on $k$ in \cref{eq:k-two-sided}, there is a constant
$c_4(\delta)>0$ such that
\[
    ab_\delta k^{1-\delta}
    \ge \frac{c_4(\delta)}{\eps^{2-2\delta}}.
\]
After decreasing the constant in the exponent to absorb the factor $1/225$
and the additive term $-1$, we obtain
\[
    D\ge m^{c_\delta/\eps^{2-2\delta}}
\]
for a constant $c_\delta>0$.

\medskip\noindent\emph{Cardinality.}
Finally, the number of queries is
$|\calQ|=B^k2^k$ and the number of indexed documents is $2^N=2^{kB}$.
Since $B\le m$ and $k\le m^{O_\delta(1)}$ in the stated regime,
both quantities are at most $2^{m^{O_\delta(1)}}$.  Deduplication can only
reduce the number of documents.  This proves the theorem.
\end{proof}

Because singleton-query Chamfer similarity is exactly MAX-IP, the theorem has
an immediate multi-vector consequence.

\begin{corollary}[Chamfer dimension lower bound]
\label{cor:chamfer-lower-bound}
Under the parameters of \cref{thm:main}, there are query point clouds and
document point clouds of at most $m$ unit vectors for which every
data-dependent single-vector $\eps$-approximation of all pairwise normalized
Chamfer similarities has dimension
\[
    D\ge m^{c_\delta/\eps^{2-2\delta}}.
\]
The query point clouds in the construction all have size one.
\end{corollary}

\begin{proof}
Replace each query vector $q\in\calQ$ by the singleton point cloud $\{q\}$ and
use the identity $\Chamfer(\{q\},X)=\MAXIP(q,X)$.  Then apply
\cref{thm:main}.
\end{proof}

\paragraph{The point-cloud threshold.}
The assumption $m\ge(1/\eps)^{A_\delta}$ serves a concrete purpose.  The hard
DNF has $s\le k^C$ terms and the document cloud contains one vector for each
term and each of its $B^W$ block alignments.  To obtain a block size
$B=m^{\Omega_\delta(1)}$, the available point-cloud budget must dominate the
formula size by a polynomial factor.  Our proof enforces the convenient
condition $m\ge s^2+1$.  No optimization of $A_\delta$ is attempted here.
Whether the point-cloud threshold can be reduced, perhaps to the
$m\gtrsim1/\eps^2$ regime, remains open.

\section*{Acknowledgements}
The proof was first obtained using a fully automated Gemini-based agentic system developed internally at Google. The authors have verified the proof and edited it for clarity of presentation, and take responsibility for the final version.

\bibliographystyle{alpha}
\bibliography{references}

\end{document}